\documentclass{article}

\usepackage{amsmath,amssymb,upgreek,natbib,enumerate,graphicx,ifpdf}
\ifpdf
\DeclareGraphicsExtensions{.pdf}
\else
\DeclareGraphicsExtensions{.eps}
\fi

\author{A. Hanyga 
and M. Seredy\'{n}ska\\
Institute of Fundamental Technological Research\\
ul. Bitwy Warszawskiej 1920 r. 14 m. 52, 02-366 Warszawa}

\title{Spatially fractional-order viscoelasticity, non-locality and a new kind of anisotropy}

\newtheorem{theorem}{Theorem}[section]

\newtheorem{lemma}[theorem]{Lemma}

\newtheorem{assumption}{Assumption}
\newenvironment{proof}{\textbf{Proof.}}{\hfill$\Box$\\ }

\newcommand{\re}{\mathrm{Re}}
\newcommand{\im}{\mathrm{Im}}
\newcommand{\Ai}{\mathrm{Ai}}
\newcommand{\kk}{\mathbf{k}}
\newcommand{\yy}{\mathbf{y}}

\newcommand{\trace}{\mathrm{trace\,}}
\newcommand{\dd}{\mathrm{d}}
\newcommand{\D}{\mathrm{D}}
\newcommand{\I}{\mathrm{I}}
\newcommand{\ii}{\mathrm{i}}
\newcommand{\e}{\mathrm{e}}
\newcommand{\Div}{\mathrm{div}\,}
\newcommand{\Grad}{\mathrm{grad}\,}

\newcommand{\KK}{\mathbf{K}}
\newcommand{\QQ}{\mathbf{Q}}
\newcommand{\B}{\mathbf{C}}
\newcommand{\ee}{\mathbf{e}}
\newcommand{\uu}{\mathbf{u}}

\begin{document}
\maketitle
\begin{abstract}
Spatial non-locality of space-fractional viscoelastic equations of motion is studied. Relaxation effects are accounted for by replacing second-order time derivatives by lower-order fractional derivatives and their 
generalizations. It is shown that space-fractional equations of motion of an order strictly less than 2 allow 
for a new kind anisotropy, associated with angular dependence of non-local interactions between stress and strain 
at different material points. Constitutive equations of such viscoelastic media are determined. Explicit 
fundamental solutions of the Cauchy problem are constructed for some cases isotropic and anisotropic non-locality.  
\end{abstract}

\noindent\textbf{Keywords.} viscoelasticity, non-local, fractional, anisotropic, bio-tissue, anomalous diffusion.
\noindent\textbf{MSC Class:} 74D05, 74A20, 74J10. 

\noindent\textbf{Notation.}\\
$(f \ast_t \, g)(t) := \int_0^t f(s)\, g(t-s) \, \dd s$\\
$(f \ast_x \, g)(x) := \int_{\mathbb{R}^3} f(x-y) \, g(y) \, \dd y$\\
Fourier transform: $\hat{f}(\kk) = \int_{\mathbb{R}^d} \, \e^{-\ii \kk\cdot x}\, f(x)\, \dd_d x$;\\
Laplace transform: $\tilde{f}(p) = \int_0^\infty \e^{-p t} \, f(t) \, \dd t $.

\section{Introduction}

Partial differential equations with fractional derivatives play an important role in modeling anomalous 
diffusion and wave propagation in media with complex structure. Equations invariant with respect to space 
and time scaling are of particular interest. While space-time fractional equations have been widely accepted 
and studied for modeling anomalous diffusion \cite{TimeSpaceFrac}, no analogous equations have been obtained 
for wave propagation in continuous media. 

Viscoelastic models considered here have two new aspects. In the first place the stress-strain constitutive
equation is non-local - the stress at a material point depends on the strain at neighboring material points - 
and as a result the  operator in the equation of motion acting on the spatial variables is non-local. In 
addition, the time derivative can be replaced by a non-local operator to allow for time-delayed response of 
stress to strain. The latter addition is however familiar in contemporary viscoelasticity. 

The other aspect of spatial non-locality is the appearance of a new kind of anisotropy. In linear viscoelasticity 
elastic anisotropy manifests itself in connection with the tensorial character of the constitutive equation 
linking stress to the strain tensor. Elastic anisotropy is associated with the symmetry properties of the 
stiffness tensor with respect to the linear transformations of the coordinates. Elastic anisotropy can result 
from crystallographic structure, fine layered structure, parallel cracks etc. It accounts for angular dependence 
of elastic and viscoelastic moduli, that the dependence of stress on the direction of stretching or shear. 

In the theory presented below there is room for a different kind of anisotropy associated with non-locality of 
the stress-strain constitutive equation. In the equations of motion the operator acting on the spatial variables 
is a pseudo-differential operator of an order $\leq 2$. If the order of the spatial operator is strictly lower 
than 2, then the spatial operator operator is non-local. This brings about a new kind of anisotropy. 
The new anisotropy is not related to the tensorial properties of the stiffness coefficients but is instead  
defined by the micro-local structure of the spatial operator. Anisotropy associated with non-locality accounts for 
the angular dependence of sensitivity of stress to the strain in the adjacent material. Such a dependence can be 
expected to be much stronger along a ligament than transversally to it. 

In the Fourier-transformed  spatial operator the anisotropy associated with non-locality is expressed in terms 
of anisotropic wave number dependence. We shall only consider a special class of such operators whose symbols 
are expressed in terms of a probability distribution of over all the directions on the unit sphere. A distribution 
over directions on a unit sphere can be expressed as an infinite series of Legendre polynomials. 
This is far more general than the concept of anisotropy considered in elasticity. 

In linear elasticity all the anisotropy classes are derived from the symmetries of the stiffness coefficients. 
The anisotropic properties of such media can be expressed in terms of tensors of finite rank. In local linear 
viscoelasticity stiffness coefficients are frequency dependent. Such understanding of anisotropy has been 
motivated by crystallographic, layered or densely cracked structures of many materials. In modelling of 
high-resolution MRI scans of anisotropic bio-tissues requires diffusion tensors of very high orders 
\cite{OzarslanMareci03}. This is due to a high complexity of cell and tissue structure. Anisotropy of 
viscoelastic wave motion in a bio-tissue can be expected to reveal the same level of complexity. Anisotropic 
diffusion is possible only for anomalous diffusion of the space-fractional type. Anomalous diffusion has a 
counterpart in non-local viscoelasticity of the kind considered in this paper.

The anisotropy of non-locality can coexist with the anisotropy associated with the symmetries of the stiffness 
coefficients because the stiffness coefficients depend on the wave number directions. This kind of anisotropy 
is however also possible in scalar models of viscoelasticity. It is interesting that for this kind of anisotropy 
explicit solutions of quite general anisotropic equations can be constructed, which is not the case for anisotropy 
associated with the symmetries of the stiffness coefficients.

In the case of anisotropy associated with the symmetries of the stiffness coefficients the simplest geometrical 
element is a plane of symmetry (in transversal symmetry). Planes of symmetry can be associated with micro-layered 
structures. Anisotropy associated with non-locality allows for different kinds of microstructure, such as 
curvilinear ligaments and channels in bio-tissues. Non-locality is also very likely in bio-tissues due to their 
complex structure and composition. 

In this paper we focus on spatial pdo's whose symbols are homogeneous functions of the wavenumber. In the 
absence of anisotropy such operators reduce to fractional-order Riesz derivatives \cite{Samko} or, in other words, 
to fractional-order Laplacians.

The usual Laplace-Fourier representations are inconvenient for the construction and analysis of the solutions 
of equations which are of fractional order with respect to both spatial variables and time variable. This leads 
to some difficulties in the analysis of the attenuation and dispersion. As discovered in 
\cite{MainardiLuchkoPagnini}, the solutions  of such equations can be expressed in terms of Mellin convolutions. 
The solution is a superposition of copies of a fixed waveform subject to a varying scaling. The above-mentioned  
representations of the solutions are fairly explicit and involve only one non-elementary function - the Wright 
function. The Wright function corrects for the difference between the orders of the time derivative and the 
spatial operator. The effect of the non-local spatial operator is represented by an elementary function provided 
the the order of the time derivative matches the order of the Laplacian.

\section{Space-time fractional wave equations}

We shall derive an explicit formula for the Green's function of an anisotropic space-time fractional 
linear viscoelastic equation, defined as the solution of the Cauchy problem
\begin{equation} \label{eq:motion} 
\rho \, \D^\beta u = Q \, u, \qquad u(0,x) = \delta(x), \quad \D u(0,x) = 0 \quad x \in \mathbb{R}^3
\end{equation}
where $\D^\beta$ denotes the Caputo derivative, $1 < \beta \leq 2$ and $Q$ is a pseudo-differential operator 
defined by its symbol 
\begin{equation} \label{eq:N}
\hat{Q}(\kk) = -\int_{\mathcal{S}} \vert \kk\cdot\yy \vert^\alpha \, \mu(\dd \yy)
\end{equation}
and $\beta \leq \alpha \leq 2$. The integral in eq.~\eqref{eq:N} extends over a unit sphere 
$\mathcal{S}$, defined by the equation $\vert \yy \vert = 1$ and the measure $\mu$ on 
$\mathcal{S}$ is non-negative and has finite mass.

Note that $\hat{Q}(\kk) = \vert \kk \vert^\alpha \, \hat{Q}(\hat{\kk})$. Hence the order $\alpha$ of the
operator is independent of the anisotropic properties of $Q$, represented by the measure $\mu$. 

Choosing $\mu$ to be a homogeneous distribution over the sphere
$$\mu(\dd \yy) = M\, \frac{\alpha+1}{4 \uppi}\sin(\theta) \, \dd \theta \,\dd \varphi$$
where $\theta$, $\varphi$ are polar coordinates on $\mathcal{S}$ we get 
\begin{equation}
\hat{Q}(\kk) = -M\, \vert \kk \vert^\alpha
\end{equation}
and thus $Q$ essentially reduces to a fractional-order Laplacian
\begin{equation}
Q = -M\, \left(-\nabla^2\right)^{\alpha/2}
\end{equation}

For $\alpha = 2$ we have $\vert \kk\cdot\yy\vert^2 = \kk\cdot \yy \yy^\mathsf{T}\, \yy$ and
thus $\hat{Q}(\kk) = -\kk \cdot \mathbf{M}\, \kk$, where $\mathbf{M}$ is a positive definite matrix 
$$\mathbf{M} := \int_\mathbf{S} \yy \yy^\mathsf{T} \, \mu(\dd \yy)$$
Consequently $Q = \nabla \cdot \mathbf{M} \, \nabla$ is an anisotropic generalization of the Laplacian.
A spatial operator of second order is ellipsoidal allows only for ellipsoidal anisotropy. 

A different anisotropic generalization of the Laplacian is obtained by choosing a measure $\mu$ with three 
mutually orthogonal support points $\yy_1, \yy_2, \yy_3$ on the sphere $\mathcal{S}$. In this case 
$\hat{Q}(\kk) = -M\, \vert \kk\vert^\alpha \, \hat{\kk}\cdot \mathbf{A}\hat{\kk}$, where 
$\hat{\kk} := \vert\kk\vert^{-1} \, \kk$.

Let $\I^\gamma$ denote the fractional integral operator
\begin{equation}
\I^\gamma\, f(t) = \int_0^t \frac{(t-s)^{\gamma-1}}{\Gamma(\gamma)} \, f(s) \, \dd s, \quad t > 0
\end{equation}
for $\gamma > 0$. For positive integer $\beta$ the Caputo derivative is an ordinary derivative.
For positive non-integer $\beta$  Caputo derivative $\D^\beta$ of order $\beta$ is defined by the formula 
\begin{equation}
\D^\beta f = \I^{n-\beta}\, \D^n \, f
\end{equation}
where $n$ is an integer such that $n -1 < \beta < n$ \cite{PodlubnyBook}.

\section{Comparison with scalar viscoelasticity}

\subsection{Construction of fractional-order viscoelastic equations of motion.}

As it stands, equation~\eqref{eq:motion} does not have the form of a viscoelastic equation of motion. We shall prove that
it is equivalent to a generic momentum balance equation $\rho \, \D^2 \, u = \Div \D \,\upsigma$, where 
the stress $\upsigma$ is given by the constitutive equation
\begin{equation}
\upsigma = g(t)\ast_t \B \, \D\, \ee, \quad \ee = \left(\nabla u + (\nabla u)^\mathsf{T}\right)/2
\end{equation}
$g$ is a completely monotone function (appendix~\ref{app:some}) 
and $\B$ is a non-local integral operator of order $\alpha-2$ acting on the space variables $x$.

We shall begin with transforming equation~\eqref{eq:motion} into a second order equation with respect 
to time. 

Equation~\eqref{eq:motion} implies that 
$u(t,x) - u(0,x) - t\, \dot{u}(0,x) \equiv \I^2\, \D^2 \, u = \rho^{-1}\,\I^\beta \, Q \, u$, hence
$\rho\, \D^2 \, u = \D^2\, \I^\beta\, Q \, u$. The right-hand side can be expressed in the form 
$\D \, \I^{\beta-1} \, Q \, u = \D \, \I^{\beta-1}\, Q\, \D u + \left[t^{\beta-1}/\Gamma(\beta)\right]\, 
Q \, u(0,x)$. \begin{equation} \label{eq:ewq}
\rho \D^2\, u = \D^2\, \I^\beta \, Q u
\end{equation}
The right-hand side of equation~\eqref{eq:ewq} can be transformed as follows:
\begin{multline}
\D^2\, \int^t_0 \tau^{\beta-1}/\Gamma(\beta)\, (Q u)(t-\tau) \, \dd \tau =\\
 \D t^{\beta-1}/\Gamma(\beta)\, (Q u)(0) + 
\D \, \int^t_0 \tau^{\beta-1}/\Gamma(\beta)\, (Q \D u)(t-\tau) \, \dd \tau = \\
\D \, t_+^{\beta-1}/\Gamma(\beta)\ast_t \, (Q \D u)(t) = \\
\int^t_0 (t-\tau)^{\beta-2}/\Gamma(\beta-1)\, (Q \D u) (\tau) \, \dd \tau 
\end{multline} 
where we have taken advantage of the inequality $\beta > 1$ and the commutation of $Q$ and $\D$ and set $(\D u)(0,x) = 0$. 
We thus have transformed equation~\eqref{eq:motion} into the following form
\begin{equation} \label{eq:ewq2}
\rho \D^2\, u = g\ast_t \, Q u
\end{equation}
where the relaxation modulus $g(t) = t_+^{\beta-2}/\Gamma(\beta-1)$. 
The relaxation modulus $g$ is locally integrable completely monotone (LICM, see Appendix~\ref{app:some}. 

In the next step we shall recast the equation in the form of a momentum conservation equation by 
setting $Q u = \Div \mathbf{\KK}(u)$. To this effect we note that $\hat{Q}$, given by equation~\eqref{eq:N},
is a homogeneous and differentiable function of degree $\alpha$:
$$ \hat{Q}(\lambda \kk) = \lambda^\alpha \, \hat{Q}(\kk)$$
Hence
$\hat{Q}(\kk) = \ii \kk \cdot \hat{\KK}(\kk)$, where $\hat{\KK} = \nabla_\kk \, \hat{Q}(\kk)/(\ii \alpha)$. We now define 
the operator $\KK$ by its symbol $\hat{\KK}$. 
The stress (momentum flux) can now be identified as $\upsigma = g\ast_t \D \, \KK(u)$ and equation~\eqref{eq:ewq2}
assumes the form of a momentum conservation equation:
\begin{equation} \label{eq:momentum}
\rho\, \D^2 \, u = \Div \upsigma
\end{equation}

We now verify whether the stress is a functional of the strain rate $\ee$. To this effect we note that 
$\hat{Q}(\kk) = -\kk\cdot \B(\kk) \, \kk$, where $\B(\kk) := -\nabla_\kk \, 
\nabla_\kk \hat{Q}(\kk)/[\alpha\, (\alpha-1)]$.
Hence $\hat{\KK}(\kk) = \ii\, \B(\kk)\, \kk/(\alpha-1)$, $\KK(u) = \B(-\ii \nabla) u$ and
$\upsigma = g\ast_t\, \D \, \B(-\ii \nabla) u = g \ast_t\, \QQ(\D \ee)$, where $\QQ$ is a non-local operator acting on 
spatial variables. Consequently stress is given as  a non-local linear functional of the strain history.
$\B$ is the extension of the stiffness tensor to the spatially non-local viscoelasticity.

The above procedure can be readily extended to vectorial equations of viscoelasticity. In this case $Q$ is a 
tensor-valued operator
\begin{equation}
(\QQ \uu)_k = Q_{kl} \, u_l
\end{equation}
where $\uu$ denotes the displacement vector, and
$$ \upsigma_{kl} = g\ast_t\, \QQ_{klmn}(\D u_{m,n})$$
We now require that (1) the stress effectively depends only on strain $\ee = 
\left(\nabla \uu + (\nabla \uu)^\mathsf{T}\right)/2$, (2) symmetry of stress $\upsigma = \upsigma^{\mathsf{T}}$
follows from directly from the constitutive equations. The two requirements will be satisfied if 
$\B_{klmn} = \B_{klnm} = \B_{lkmn}$. The above symmetries are satisfied if there is a homogeneous generating function
$V(\kk)$ of degree $\alpha+2$, with continuous derivatives up to 4-th order, such that 
\begin{equation}
Q_{kl}(\kk) = \frac{\partial^2 V(\kk))}{\partial k_k \, \partial k_l}
\end{equation}

For the operator defined by \eqref{eq:N} the generating function assumes the following form 
\begin{equation}
V(\kk) = - \frac{1}{(\alpha + 1)\, (\alpha + 2)} \int_{\mathcal{S}} \vert \kk\cdot\yy\vert^{\alpha+2} 
\, \mu(\dd \yy) 
\end{equation}

For simplicity we have assumed here that the kernel $g$ is scalar. Otherwise we would have to take 
into account commutation of the tensor-valued functions or else to consider convolution with a tensor-valued
kernel dependent on both time and space variables. Viscoelastic equations with tensor-valued relaxation modules 
$g$ and local spatial operators are considered in \cite{HanDuality}. 

\subsection{Energy conservation and its implications.}

Existence of an energy conservation with a non-negative energy functional plays an important role in
proofs of well-posedness of the equations. We shall examine some assumptions on the operator $Q$ 
that ensure existence of a non-negative energy 
satisfying a conservation equation. 

The following energy balance holds for solutions of equation~\eqref{eq:motion} $\beta = 2$ with 
$Q = \Div \circ \B \circ \Grad$:
\begin{equation} \label{eq:energy-space}
\frac{\dd}{\dd t} \left[ \int_{\mathbb{R}^3}  \frac{1}{2} \rho \, \dot{u}^2 \, \dd x + 
\frac{1}{2} \int_{\mathbb{R}^3} (\nabla u)^\mathsf{T}\, \B \, (\nabla u)\, \dd x \right] = 0
\end{equation}
where $\dot{u} := \D u$. The pseudo-differential operator $\B$ can be expressed as a convolution with a spatial 
distribution $H$: $\B f = H\ast_x\, f$. The Fourier transform $\hat{H}$ of the kernel 
$H$ is the symbol of the operator $\B$. The second term in the square brackets on the left-hand side represents the stored energy $U(t)$ at time $t$.

We now impose the condition that the stored energy is non-negative. Applying the Fourier transformation to $U$ 
we have on account of Plancherel's theorem
\begin{equation}
U(0) = \frac{1}{2} \int_{\mathbb{R}^3} (\ii \kk \hat{u}(0,\kk))^\dag\, \hat{H}(\kk) \, (\ii \kk \hat{u}(0,\kk)) \,
\dd \kk \equiv -\frac{1}{2} \int_{\mathbb{R}^3} \hat{u}(0,\kk)^\dag\, \hat{Q}(\kk)\,
\hat{u}(0,\kk)\, \dd \kk 
\end{equation}
In view of the arbitrariness of the function $u(0,x)$ we shall impose the condition that $\hat{Q}(\kk)$ is
Hermitian positive semidefinite
\begin{equation}
\hat{Q}(\kk) \leq 0
\end{equation}
i.e. $\mathbf{w}^\dag\, \hat{Q}(\kk) \, \mathbf{w}$ for all complex vectors $\mathbf{w} \in \mathbb{C}^3$. 

In particular, for scalar equations with $Q = -M\, \left(-\nabla^2\right)^{\alpha/2}$, $1 < \alpha \leq 2$, we have
$\hat{Q}(\kk) = - M \,\left(\kk^2\right)^{\alpha/2} \leq 0$, and 
\begin{equation}
\hat{H}(\kk) = -\frac{1}{\alpha\, (\alpha-1)} \frac{\partial^2\, \left(\kk^2\right)^{\alpha/2}}{\partial \kk^2} = 
-\alpha\,\left(\mathbf{I} -(2-\alpha)\,  \hat{\kk}\, \hat{\kk} \right)\,\left(\kk^2\right)^{\alpha-2}
 \leq 0
\end{equation} 
because $2 - \alpha \leq 1$. 

Let us now turn to $1 < \beta < 2$. The equation $\rho\, \D^\beta \, u = Q u$ can be transformed to the more familiar form 

If the stress is defined by the equation
\begin{equation} \label{eq:stress}
\upsigma = g \ast_t H \ast_x\, \nabla \D u
\end{equation}
then equation~\eqref{eq:ewq2} assumes the familiar form \eqref{eq:momentum}. Taking the scalar product of
both sides of the last equation with the vector $\D u$ we arrive at the formula
\begin{equation}
\frac{\dd}{\dd t} \frac{\rho}{2} (\D u)^2 = -\trace (\upsigma \, \nabla \D u)
\end{equation} 
We shall now construct the stored energy functional $U$ in such a way that $\dd U/\dd t = \trace (\upsigma \, \nabla \D u)$. 

The function $g(t)$ is completely monotone and locally integrable (LICM). LICM functions 
are causally positive definite (CPD) \cite{GripenbergLondenStaffans}. We shall therefore carry out our analysis for a general 
causally positive definite function $G$. A CPD function $g$ can be expressed as the Fourier transform of of a positive
Radon measure $m$
\begin{equation}
g(t) = \int_{-\infty}^\infty \e^{\ii s t} m(\dd s) 
\end{equation}
where
\begin{equation}
m(\dd s)  = \frac{1}{\uppi} \Re \hat{g}(s)
\end{equation}
In particular, for $g(t) = t_+^{\beta-2}/\Gamma(\beta-1)$ we have 
$m(\dd s) = (1/\uppi)\, \sin(\beta \uppi/2) \, \vert s \vert^{1-\beta}\, \dd s$ \cite{Gelfand}. 

Let $\psi := H\ast_x\, \nabla \D u$, 
$$\yy(t,x;s) := \int_0^t \e^{\ii s (t-\tau)} \, \psi(\tau) \, \dd \tau$$
We now make the following assumption 
\begin{assumption} \label{ass:1}
For each $\kk$ the matrix 
$\hat{H}(\kk)$ is positive definite and symmetric.
\end{assumption} 
The function $\hat{H}$ an even function of $\kk$ because $H$ is real-valued.
Define the kernel 
\begin{equation}
H_{-1}(x) = \frac{1}{2 \uppi} \int_{\mathbb{R}^3} \,\e^{-\ii \kk\cdot x} \hat{H}{\kk}^{-1} \, \dd \kk
\end{equation}
It is easy to check that $H_1(x) \ast_x \, H(x) = \delta(x)$. We now define the stored energy functional
\begin{equation} \label{eq:stored}
U = \frac{1}{2} \int_{\mathbb{R}^3}  \int_{\mathbb{R}^3} \int_{-\infty}^\infty \yy(t,x;s)^\dag\, H_1(x-y) \, \yy(t,y;s) 
 \, m(\dd s)\, \dd y \, \dd y 
\end{equation}
Since $\D \yy = \psi - \ii s \yy$ and $H_{-1}$ is even and symmetric
\begin{multline}
\frac{\dd U}{\dd t} = \int_{-\infty}^\infty \frac{1}{2} \int_{\mathbb{R}^3}  \int_{\mathbb{R}^3}  
\left( \psi(t,x)^\dag \, H_{-1}(x-y)\, \yy(t,x;s) + \yy^\dag\, H_{-1}(x-y) \, \psi(t,x)  \right) \,
 \dd y\,  \dd x\,
 m(\dd s) = \\ \int_{-\infty}^\infty  \left\{ \int_{-\infty}^\infty  \left[\int_{-\infty}^\infty \yy(t,x;s) \, m(\dd s)\right]
\, \int_{-\infty}^\infty H_{-1}(x - y) \, \psi(t,y) \, \dd y  \right\} \, \dd x \\ = \int_{\mathbb{R}^3} 
\trace(\upsigma(t,x) \, \nabla \D u(t,x)) \, \dd x
\end{multline}
as required.

We have thus proved energy conservation 
\begin{equation}
\frac{\dd}{\dd t} \left[ \int_{\mathbb{R}^3} \frac{\rho \, (\D u)^2}{2} \, \dd x + U\right] = 0
\end{equation}
with the stored energy $U$ defined by equation~\eqref{eq:stored}. 
Assumption~\ref{ass:1} played a key role in the construction of the stored energy potential. 

Assumption~\ref{ass:1} also ensures that $U \geq 0$. Indeed, by Plancherel's theorem,
$$U = \int_{\mathbb{R}^3} \hat{\yy}(t,\kk;s)^\dag\, \hat{H}(\kk)^{-1}\, \hat{\yy}(t,\kk;s) \, m(\dd s) \geq 0$$

\section{Solution of the Cauchy problem for $1 < \beta \leq \alpha \leq 2$.}

\subsection{Formulation of the Cauchy problem.}

Consider an abstract time-fractional equation 
\begin{equation}\label{eq:abstract}
\D^\beta \, u = A u
\end{equation}
where $A$ is an operator and $\D^\beta$ is the Caputo fractional derivative.
For $\beta \leq 1$  the Caputo derivative $\D^\beta = \I^{1-\beta}\, \D$ and, using the semigroup property of 
the fractional integral operators $\I^\alpha$,  equation~\eqref{eq:abstract} is equivalent to the integral equation
$$ \I \, \D u = \I^\beta \, A \, u$$
or, more explicitly,
$$u(t) = u(0) + \I^\beta \, A \, u$$
We thus expect that equation~\eqref{eq:abstract} with the initial condition $u(0) = u_0$ has a unique solution. 

If $1 < \beta \leq 2$, then $\D^\beta = \I^{2-\beta}\, \D^2$ and equation~\eqref{eq:abstract} 
is equivalent to $\I^2\, \D^2 u = \I^\beta \, A \, u$, or
$$u(t) = u(0) + t\, u^\prime(0) + \I^\beta \, A \, u$$
where $u^\prime = \D u$. 
In this case we expect that equation~\eqref{eq:abstract} with the initial conditions 
\begin{equation} \label{eq:IV}
u(0) = u_0, \quad u^\prime(0) = \dot{u}_0
\end{equation}
has a unique solution. 

Applying the Laplace transformation we have
\begin{equation}
\tilde{u}(p) = \frac{u_0 \, p^{\beta-1}}{p^\beta + \tilde{f}(p)}
\end{equation}
for $0 < \beta \leq 1$ and
\begin{equation}
\tilde{u}(p) = \frac{u_0 \, p^{\beta-1} + \dot{u}_0 \,p^{\beta-2} }{p^\beta + \tilde{f}(p)}
\end{equation}
for $1 < \beta \leq 2$.

\subsection{Solution of the Cauchy problem.}

The Cauchy problem for the scalar equation \eqref{eq:motion}, \eqref{eq:N} in 3 dimensions has an 
explicit solution 
in terms of an integral over a unit sphere with an integrand involving Wright functions 
\cite{MainardiMLWright2011,MainardiMuraPagnini10,GorenfloLuchkoMainardi99}. The Wright functions are numerically computable by an integral representation 
\cite{LuchkoWright,LuchkoTrujilloVelasco10}. An alternative method for numerical integration of would involve pseudo-spectral methods combined with a smart representation of the 
Caputo fractional derivative 
the method of \cite{YuanAgrawal,LuHanygaJCP,Diethelm2010}. On the other hand, for some values of the index the Wright function can be expressed in terms of the exponential and the Airy function.

Applying to equation~\eqref{eq:motion} the Laplace transform $f(t) \rightarrow \mathcal{L}[f](p) \equiv \tilde{f}(p) 
= \int_0^\infty \e^{-p t}\, f(t)\, \dd t$, the Fourier transform 
$g(x) \rightarrow \mathcal{F}[g](\kk) \equiv \hat{g}(\kk) = \int_{\mathbb{R}^3} \e^{-\ii \kk\cdot x}\, g(x) \, \dd x$ 
and the identity 
\cite{PodlubnyBook}
\begin{equation}
\mathcal{L}\left[\D^\beta f\right](p) = p^\beta \, \tilde{f}(p) - p^{\beta-1}\, f(0) - p^{\beta-2}\, f^\prime(0)
\end{equation}
valid for $1 < \beta \leq 2$, we obtain the equation
\begin{equation}
U(p,\kk) := \mathcal{F}\left[\mathcal{L}[u]\right](p,\kk) = \widehat{u_{0}}(\kk)\, \frac{p^{\beta-1}}{p^\beta + \hat{Q}(\kk)} 
+ \widehat{\dot{u}_{0}}(\kk)\, \frac{p^{\beta-2}}{p^\beta + \hat{Q}(\kk)}
\end{equation}

The solution of the Cauchy problem \eqref{eq:motion}, \eqref{eq:IV} will be obtained by inversion of the Laplace and Fourier transformations,
$$u(t,x) = \frac{1}{2 \uppi \ii} \int_\mathcal{B} \dd p \, \e^{p t} \frac{1}{(2 \uppi)^3} \int_{\mathbb{R}^3} \dd \kk \,
\e^{\ii \kk\cdot x} \, U(p,\kk)$$
Hence 
\begin{equation}
u(t,x) = \int_{\mathbb{R}^3}  G_{\beta,Q}(t,x-y)\, u_0(y) \, \dd y + \int_{\mathbb{R}^3}  H_{\beta,Q}(t,x-y)\, \dot{u}_0(y) \, \dd y
\end{equation}
where 
\begin{equation} \label{eq:GN}
G_{\beta,Q}(t,x) = \frac{1}{2 \uppi \ii} \int_\mathcal{B} \dd p \, \e^{p t} \frac{1}{(2 \uppi)^3} \int_{\mathbb{R}^3}  
\e^{\ii \kk\cdot x} \, \frac{p^{\beta-1}}{p^\beta + \hat{Q}(\kk)} \, \dd_3 \kk
\end{equation}
\begin{equation} \label{eq:HN}
H_{\beta,Q}(t,x) = \frac{1}{2 \uppi \ii} \int_\mathcal{B} \dd p \, \e^{p t} \frac{1}{(2 \uppi)^3} \int_{\mathbb{R}^3} 
\e^{\ii \kk\cdot x} \, \frac{p^{\beta-2}}{p^\beta + \hat{Q}(\kk)} \, \dd_3 \kk
\end{equation}
and $\dd \kk = k^2 \, \dd k \times \sin(\vartheta) \, \dd \vartheta \times \dd \varphi$.

The functions $G_{\beta,Q}$ and $H_{\beta,Q}$ will be called the first fundamental solution and the second fundamental 
solution of the Cauchy problem for \eqref{eq:motion}, respectively.

Exchanging the order of the inverse Laplace and Fourier transformations we have
\begin{gather} \label{eq:Ftr}
G_{\beta,Q}(t,x) = \frac{1}{(2 \uppi)^3} \int_{\mathbb{R}^3} 
\e^{\ii \kk\cdot x} \, E_\beta(\hat{Q}(\kk)\, t)\, \dd_3 \kk\\
H_{\beta,Q}(t,x) = \frac{1}{(2 \uppi)^3} \int_{\mathbb{R}^3} 
\e^{\ii \kk\cdot x} \, E_{\beta,2}(\hat{Q}(\kk)\, t)\, \dd_3 \kk\
\end{gather}

If $\hat{Q} = -a \, \vert \kk \vert^\alpha$, $a > 0$, then $U(p,\kk)$ is a function $V(p,k)$ of $k := \vert \kk \vert$ 
and the integration 
over the angular coordinates in \eqref{eq:GN} and \eqref{eq:HN} can be carried out explicitly, yielding 
$G^{(3)}_{\gamma,\alpha} := G_{\beta,Q}$, where we use a new parameter $\gamma := \beta/\alpha$ for reasons that will soon become clear:
\begin{multline} \label{eq:1Dto3D}
G^{(3)}_{\gamma,\alpha}(t,x) = \frac{1}{2 \uppi \ii} \int_\mathcal{B} \dd p \, \e^{p t} \frac{1}{(2 \uppi)^2\, \ii \, r} 
\int_{-\infty}^\infty k \, \e^{\ii k r} V(p,\vert k \vert) \, \dd k = \\
\frac{1}{2 \uppi \ii} \int_\mathcal{B} \dd p \, \e^{p t} \frac{1}{2 \uppi\, r} \frac{\partial}{\partial r}
\frac{1}{2 \uppi} \int_{-\infty}^\infty
k \, \e^{\ii k r} V(p,\vert k \vert) \, \dd k = - \frac{1}{2 \uppi r} \frac{\partial}{\partial r} G^{(1)}_{\gamma,\alpha}(t,r)\vert_{r=\vert x\vert}
\end{multline}
where $G^{(1)}_{\gamma,\alpha}$ denotes the solution of the same problem in one-dimensional space. 
Similarly
\begin{equation}
H^{(3)}_{\gamma,\alpha}(t,x) = - \frac{1}{2 \uppi r} \frac{\partial}{\partial r} H^{(1)}_{\gamma,\alpha}(t,r)\vert_{r=\vert x\vert}
\end{equation}

Let $f$ be an even differentiable function on the real line.
The mapping $f(r) \rightarrow F(r) := -f^\prime(r)/(2 \uppi r)$ has two noteworthy properties:
\begin{enumerate}[(i)]
\item if $f$ vanishes at infinity, then $$\int_{\mathbb{R}^3} f(\vert x\vert) \, \dd x = \int_{-\infty}^\infty f(r) \, \dd r ; $$
\item if $f$ is unimodal with a mode at 0 (i.e. if $f^\prime(r) \leq 0$ for $r > 0$), then $F \geq 0$.
\end{enumerate}
Consequently, if $f$ is a probability density on the real line with a maximum or a singularity at 0, then $F(\vert x\vert)$
is a probability density in $\mathbb{R}^3$. This observation is crucial for distinguishing between diffusion and wave propagation in 1 and 3 dimensions.

The scaling transformations allow considerable simplification of the problem by reducing the number of unknowns. 
Changing the integration
variables $p \rightarrow s = p t$, $k \rightarrow \kappa = k\, t^\gamma$ we get the scaling relation
\begin{gather}
G^{(d)}_{\gamma,\alpha}(t,x) = t^{-d \gamma}\, G^{(d)}_{\gamma,\alpha}\left(1,x/t^{\gamma}\right)\\
H^{(d)}_{\gamma,\alpha}(t,x) = t^{1-d \gamma}\, H^{(d)}_{\gamma,\alpha}\left(1,x/t^{\gamma}\right)
\end{gather}

A simple integral representation of the $G^{(1)}_{\gamma,\alpha}(1,x)$ is constructed in the paper of Mainardi 
et al. \cite{MainardiLuchkoPagnini}:
\begin{equation}
G^{(1)}_{\gamma,\alpha}(1,x) = \int_0^\infty M_\gamma(\xi)\, X_\alpha(x/\xi) \, \frac{\dd \xi}{\xi}
\end{equation}
where 
\begin{equation}
X_\alpha(y) := \begin{cases} \frac{1}{\uppi} \frac{\vert y\vert^{\alpha-1} \, 
\sin\left(\alpha \uppi/2\right)}{1 + 2 \vert y\vert^\alpha \, \cos(\alpha \uppi/2) + \vert y\vert^{2 \alpha}}\quad \text{for $0 < \alpha < 2$}\\
(\delta(y-1) + \delta(y+1))/2 \quad \text{for $\alpha = 2$}
\end{cases}
\end{equation}
We derive it in a slightly different way in Appendix~\ref{app:identity} and Appendix~\ref{app:X}.
It is easy to see that $X_\alpha(y) \geq 0$ for $-\infty < y < \infty$ and 
$$\int_{-\infty}^\infty X_\alpha(y)\, \dd y = 1$$
hence $X_\alpha$ is a probability distribution.
$X_\alpha$ is the function denoted by $N^0_\alpha$ in \cite{MainardiLuchkoPagnini}. The function 
$M_\gamma$ is a special case of th
Wright function (Appendix~\ref{app:Wright}):
\begin{equation}
M_\gamma(z) = W_{\gamma,1-\gamma}(-z), \qquad 0 < \gamma < 1
\end{equation}
Consequently 
\begin{equation}
G_{\gamma,\alpha}^{(3)}(1,x) = U^{(\gamma,\alpha)}(r) := \int_0^\infty M_{\gamma}(\xi)\, X^{(3)}_\alpha(r/\xi) \, \frac{\dd \xi}{\xi}
\end{equation}
where $r = \vert x \vert$ and $\gamma = \beta/\alpha$, while 
\begin{equation}
X^{(3)}_\alpha(y) := -\frac{1}{2 \uppi^2\, y} \frac{\partial}{\partial y} \frac{y^{\alpha-1} \, 
\sin\left(\alpha \uppi/2\right)}{1 + 2 y^\alpha \, \cos(\alpha \uppi/2) + y^{2 \alpha}}, \qquad y > 0,\quad 0 < \alpha < 2
\end{equation}
It is easy to see that $X_\alpha^{(3)}(y) = y^{\alpha-3}\, Z_\alpha(y)$, where $Z_\alpha$ is regular at 0. 

In Fig.~\ref{fig:Gplot1} scaled plots of $Y_\alpha$ are shown. The function $X_1$ is the Cauchy probability density
and $X_1^{(3)}$ is non-negative. For $\alpha \leq 1$ the function $X_\alpha$ is unimodular with a singularity at 
$y = 0$ and therefore $X^{(3)}_\alpha$ is non-negative. For $\alpha > 1$ the maximum of $X_\alpha$ is shifted to the 
right of $y = 0$ and therefore the function $X^{(3)}_\alpha$ changes sign.  

\begin{figure}
\includegraphics[width=\linewidth]{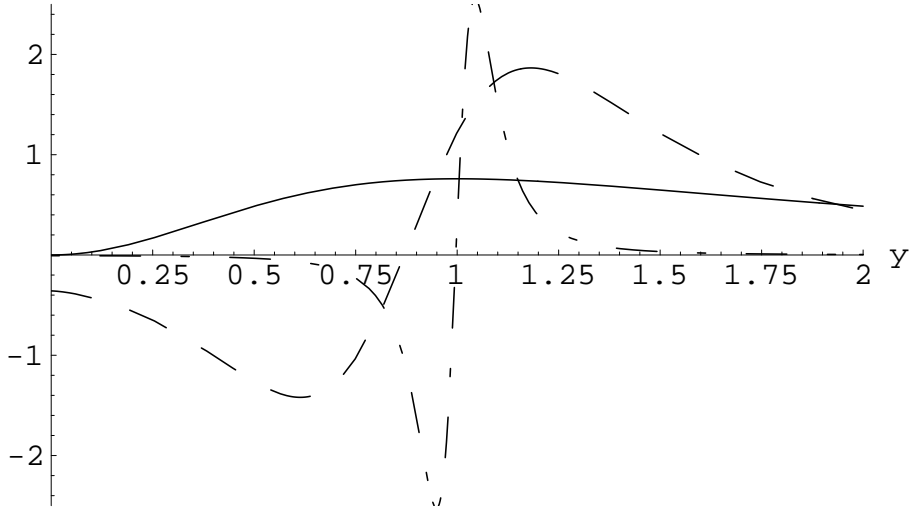}
\caption{Scaled plots of the function $Z_\alpha(y) := y^{3-\alpha}\, X^{(3)}_\alpha(y)$: (1) solid line:  $30\, Z_{1}$,
(2) dashed: $20 \, Z_{1.5}$, (3) dot-dashed: $Z_{1.9}$.} 
\label{fig:Gplot1}
\end{figure}

The parameter $\xi$ scales the spatial coordinates of the component waves $Y_\alpha$. 
The function $M_\gamma(\xi)$ represents the weights of various scalings of the
solution.
In Fig.~\ref{fig:Wright} the function $M_\gamma(z)$ is plotted for $\gamma = 1/2, 1/3, 2,3$. 
The shift of the maximum of the Wright function for $\gamma > 1/2$ is noteworthy.
\begin{figure}
\includegraphics[width=\linewidth]{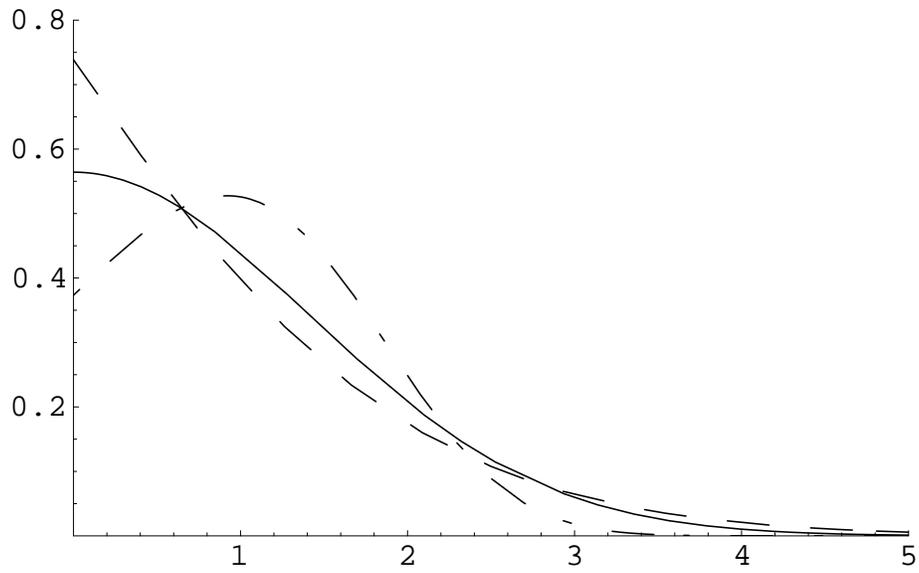}
\caption{The function $M_\gamma(z)$: (1) solid line: $\gamma = 1/2$, (2) dashed line: $\gamma = 1/3$; 
(3) dot-dashed line: $\gamma = 2/3$ .} 
\label{fig:Wright}
\end{figure}

The functions $G^{(d)}_{(2/3,\alpha}$ for $d = 1,3$ and $\alpha = 1.5$ and $1.9$ are shown in Fig.~\ref{fig:u23-1D} and 
Fig.~\ref{fig:u23-3D}. The values of $\beta$ are 1, and 1.2666. 
\begin{figure}
\includegraphics[width=\linewidth]{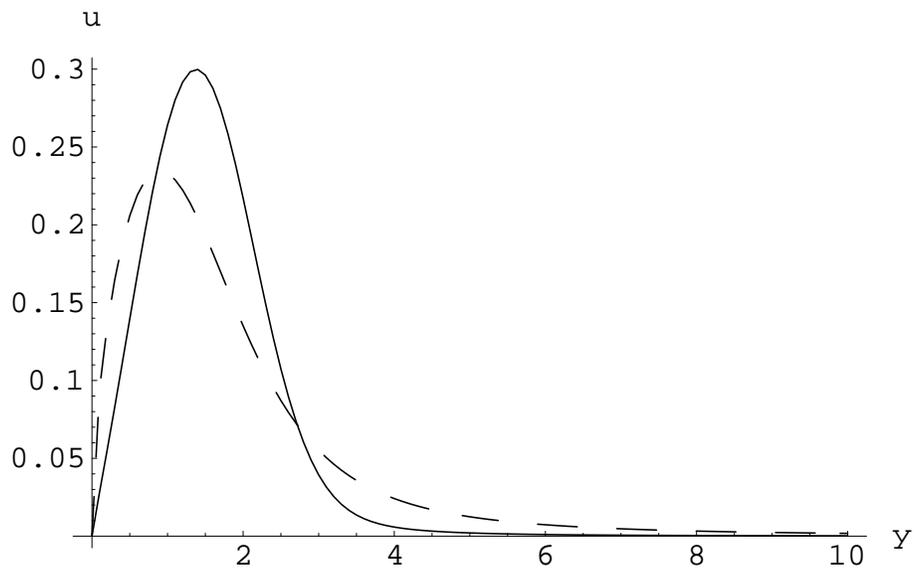}
\caption{The function $G^{(1)}_{(2/3,\alpha)}/2$ for $\alpha = 1.9$ (solid line) and $\alpha = 1.5$ (dashed line).}
\label{fig:u23-1D}
\end{figure}
\begin{figure}
\includegraphics[width=\linewidth]{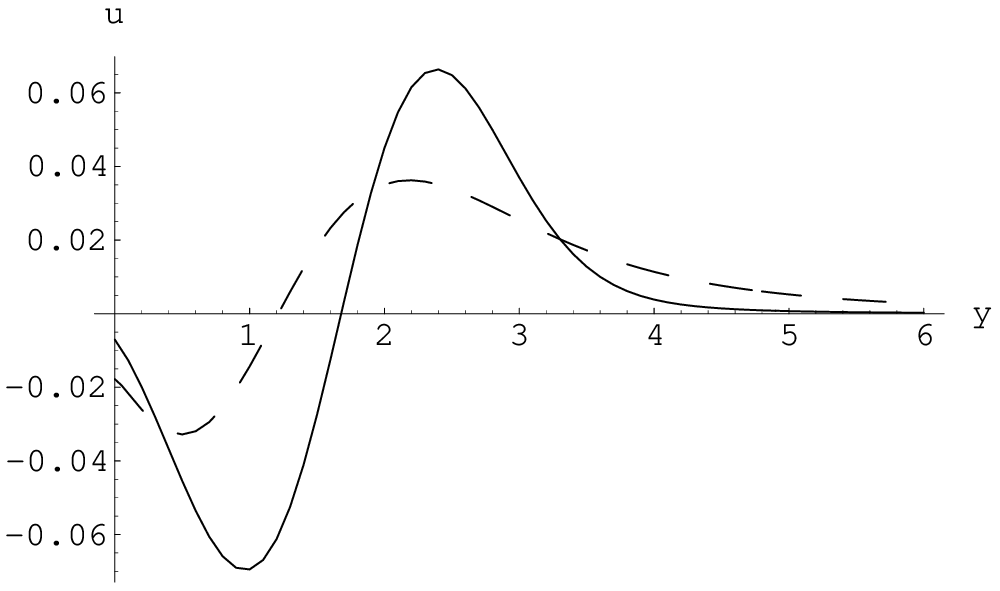}
\caption{The function $G^{(3)}_{(2/3,\alpha)}/2$ for $\alpha = 1.9$ (solid line) and $\alpha = 1.5$ (dashed line).}
\label{fig:u23-3D}
\end{figure}

\section{Fundamental solutions for anisotropic non-locality.}

\subsection{Neutral case $1 < \beta = \alpha \leq 2$.}

We now assume that $\hat{Q}(\kk)$ is given by equation~\eqref{eq:N}. The inverse Laplace transform of 
$\hat{\tilde{G}}(p,\kk)$ and $\hat{\tilde{H}}(p,\kk)$ is 
$E_\beta\left(\hat{Q}(\kk)\, t^\beta\right)$ and $E_{\beta,2}\left(\hat{Q}(\kk)\, t^\beta\right)$, respectively
or, equivalently, $E_\beta\left(-k^\alpha \, F(\hat{\kk})\, t^\beta\right)$ and $E_{\beta,2}\left(-k^\alpha \, F(\hat{\kk})\, t^\beta\right)$, respectively, 
where $k := \vert \kk \vert$, $\hat{\kk} := k^{-1}\, \kk$ and $F(\hat{\kk}) = -\hat{Q}(\hat{\kk})$. The Mittag-Leffler 
function $E_{\lambda,\mu}(z)$ is defined in the appendix and $E_\lambda := E_{\lambda,1}$.

We shall begin with the case $\beta = \alpha > 1$. 

The first fundamental solution can be obtained by inverting the Laplace and Fourier transform:
\begin{multline*}
G(t,x) = \frac{1}{(2 \uppi)^3} \int_0^\infty k^2 \, \dd k \int_{\mathcal{S}} \dd_2 \hat{\kk} \, 
E_\alpha\left(-k^\alpha \, F(\hat{\kk})\, t^\alpha\right) \, \e^{\ii \kk\cdot x} = 
\\
-\nabla^2 \, \frac{1}{(2 \uppi)^3\, t} \, \int_{\mathcal{S}} \dd_2 \hat{\kk}\, F(\hat{\kk})^{-1/\alpha} 
\int_0^\infty E_\alpha\left(-\kappa^\alpha\right)\, \e^{\ii \kappa \,\left[(\hat{\kk}\cdot \hat{x})\, r/ F(\hat{\kk})^{1/\alpha}\,t\right] } 
\, \dd \kappa =: G^{(3)}_{1,\alpha}(t,x)
\end{multline*} 
where $\hat{x} := r^{-1}\, x$.
The inverse Fourier transforms of $E_\alpha\left(-\kappa^\alpha\right)$ and 
$E_{\alpha,2}\left(-\kappa^\alpha\right)$ are $X_\alpha(y)$ and $Y_\alpha(y)$, respectively (Appendix~\ref{app:X}), hence
\begin{gather}
G^{(3)}_{1,\alpha}(t,x) = -\frac{1}{(2 \uppi)^3\, t} \int_{\mathcal{S}} \frac{\dd_2 \hat{\kk}}{F(\hat{\kk})^{1/\alpha}} \,
\nabla^2 \, X_\alpha\left(\hat{\kk}\cdot\hat{x} \frac{r}{t \, F(\hat{\kk})^{1/\alpha}}\right)\\
H^{(3)}_{1,\alpha}(t,x) = -\frac{1}{(2 \uppi)^3\, t} \int_{\mathcal{S}} \frac{\dd_2 \hat{\kk}}{F(\hat{\kk})^{1/\alpha}} \,
\nabla^2 \, Y_\alpha\left(\hat{\kk}\cdot\hat{x} \frac{r}{t \, F(\hat{\kk})^{1/\alpha}}\right)
\end{gather}
The subscript "1" refers to the value of $\gamma = \beta/\alpha$.

Choosing the spherical coordinates $(r,\vartheta,\varphi)$ in the $x$-space in such a way that $\hat{kk}$ 
corresponds to $\vartheta = 0$ and
$\hat{\kk}\cdot \hat{x} = \cos(\vartheta)$, and substituting 
$$\nabla^2 = \frac{1}{r^2} \frac{\partial}{\partial r} \left( r^2 \frac{\partial}{\partial r}\right) + 
\frac{1}{\sin(\vartheta)}  \frac{\partial}{\partial \vartheta} \left( \sin(\vartheta) 
\frac{\partial}{\partial \vartheta}\right)$$ a rather complicated but fairly explicit expression is obtained.

We shall denote the functions $G(t,x)$ and $H(t,x)$ obtained above by
$G^\mu_{1,\alpha}(t,x)$ and $H^\mu_{1,\alpha}(t,x)$.

\subsection{Anisotropic non-locality, $1 < \beta < \alpha \leq 2$.}

For the case $\beta < \alpha$ we shall use the identity
\begin{equation} \label{eq:identity}
E_\beta(-y) = \int_0^\infty M_\gamma(\xi) \, E_\alpha\left(-y\,\xi^\alpha\right)\, \dd \xi
\end{equation}
where $\gamma = \beta/\alpha$,
see Appendix~\ref{app:identity}.
Hence, setting $t = 1$ for simplicity,
$$E_\beta\left(-k^\alpha \, F(\hat{\kk}\, t^\beta) \right) = \int_0^\infty M_{\gamma}(\xi) \,
E_\alpha\left(-\xi^\alpha \, k^\alpha\,  F\left(\hat{\kk}\, \left(t^\gamma\right)^\alpha\right)\right) \, \dd \xi$$

Applying the inverse Fourier transformation we get the first fundamental solution $G^\mu_{\gamma,\alpha}$,
\begin{equation}
G^\mu_{\gamma,\alpha}(t,x) = \int_0^\infty M_\gamma(\xi)\, G^\mu_{1,\alpha}\left(t^\gamma,x/\xi\right)
\, \frac{\dd \xi}{\xi} 
\end{equation}
or, equivalently
\begin{equation}
G^\mu_{\gamma,\alpha}(t,x) = \int_0^\infty M_\gamma(\xi)\, G^\mu_{1,\alpha}\left(1,y/\xi\right)
\, \frac{\dd \xi}{\xi} 
\end{equation}
where $y = x/t^\gamma$. Similarly the second fundamental solution is given by the Mellin convolution 
\begin{equation}
H^\mu_{\gamma,\alpha}(t,x) = \int_0^\infty N_\gamma(\xi)\, G^\mu_{1,\alpha}\left(1,y/\xi\right)
\, \frac{\dd \xi}{\xi} 
\end{equation}

\section{Concluding remarks.}

Wave equations with fractional-order Laplacians and their non-local generalizations allow for a new class
of models compatible with very general anisotropy. They are compatible with the principles of viscoelasticity.

Fairly explicit expressions have been obtained for the fundamental functions of the Cauchy problem.
While solutions by the inverse Laplace and inverse Fourier transformation are intractable,
solutions in the form of a Mellin convolution are easier to obtain and analyze. Analysis of the attenuation and 
dispersion is however more difficult than in local viscoelasticity. For a space-time fractional equation 
the dispersion and attenuation are defined by an equation $\rho \, p^\beta + M \vert \kk \vert^\alpha = 0$. It is
easy to see that $\vert \kk \vert \sim_\infty\, A  \omega^\gamma$ for large $\omega := \ii p$. Thus
the attenuation grows at a sublinear rate in the high-frequency region.

\bibliography{mrabbrev,ownnew12,mathnew12}

\appendix

\section{Some definitions}
\label{app:some}

A function $f$ on the open positive real half-line $]0,\infty[$ is said to be completely monotone (CM)
if it is infinitely differentiable and $(-1)^n \, \D^n \,f(t) \geq 0$ for every non-negative integer $n$ and
$t > 0$.

A CM function can have a singularity at 0. A CM function is locally integrable (LICM) if
it its integral over the segment $[0,1]$ is finite.

A measurable real function $f$ on $[0,\infty[$ is said to be causally positive definite (CPD) if 
$$ \int_{-\infty}^\infty  \int_0^\infty f(y)\, \varphi(x-y)\, \varphi(x) \geq 0$$
for every compactly supported real test function $\varphi$.

By a theorem in \cite{GripenbergLondenStaffans} every LICM function is CPD. 
Gripenberg's extension of Bochner's theorem asserts that every continuous CPD function is equal for $x \geq 0$ 
to the Fourier transform of a finite positive Radon measure. For our purposes a the restriction to 
bounded functions is unwelcome. 
The Bochner-Schwartz theorem is an extension of Bochner's theorem to positive definite tempered distributions: 
\begin{theorem}
A causal positive definite tempered distribution $f$ is the Fourier transform of a positive tempered
Radon measure:
$$
f(x) = \int_{-\infty}^\infty \e^{-\ii \kk\cdot x} \, \mu(\dd \kk), \qquad x > 0
$$
\end{theorem}
where $\mu$ is a tempered Radon measure (i.e. a non-negative tempered distribution). This theorem captures functions with singularities typical of Gelfand's homogeneous distributions.

\section{Wright functions}
\label{app:Wright}

The Wright function depend on two parameters. They are defined in terms of power series expansion 
$$W_{\lambda,\mu}(z) := \sum_{n=0}^\infty \frac{z^n}{n! \, \Gamma(n \lambda + \mu)}, 
\qquad \lambda > -1, \; \mu \in \mathbb{C}$$
They also have an integral representation 
$$W_{\lambda,\mu}(z) = \frac{1}{2 \uppi \ii} \int_{\mathcal{H}} \e^{p - z\,p^{-\lambda}} 
\frac{\dd p}{p^\mu}, \qquad \lambda > -1, \; \mu \in \mathbb{C}$$
where $\mathcal{H}$ denotes the Hankel contour encircling the cut along the negative real axis, as 
shown in Fig.~\ref{fig:Hankel}.

The Wright functions $M_\gamma$ and $N_\gamma$ are special cases of the Wright function
\begin{gather}
M_\gamma(z) = W_{-\gamma,1-\gamma}(-z) \\
N_\gamma(z) = W_{-\gamma,2-\gamma}(-z)
\end{gather}
The best review paper on the Wright function and $M_\gamma$ is \cite{MainardiMLWright2011}. The function 
$N_\gamma$ has not been studied before. Their series expansions are
\begin{gather}
M_\gamma(z) = \sum_{n=0}^\infty \frac{(-z)^n}{n! \, \Gamma(1 - \gamma\,(n + 1)} 
= \frac{1}{\uppi} \sum_{n=1}^\infty \frac{(-z)^{n-1}}{(n-1)!} \Gamma(n\,\gamma) \, 
\sin(n\,\uppi \gamma)\\
N_\gamma(z) = \sum_{n=0}^\infty \frac{(-z)^n}{n! \, \Gamma(2 - \gamma\,(n + 1)} 
\end{gather}
These functions also have the integral representations
\begin{gather} \label{eq:M} 
M_\gamma(z) = \frac{1}{2 \uppi \ii} \int_\mathcal{H} \xi^{\gamma-1}\, \e^{\xi - z \, \xi^\gamma} \, \dd \xi \\
N_\gamma(z) = \frac{1}{2 \uppi \ii} \int_\mathcal{H} \xi^{\gamma-2}\, \e^{\xi - z \, \xi^\gamma} \, \dd \xi 
\end{gather}
\begin{figure}
\includegraphics[width=0.8\linewidth]{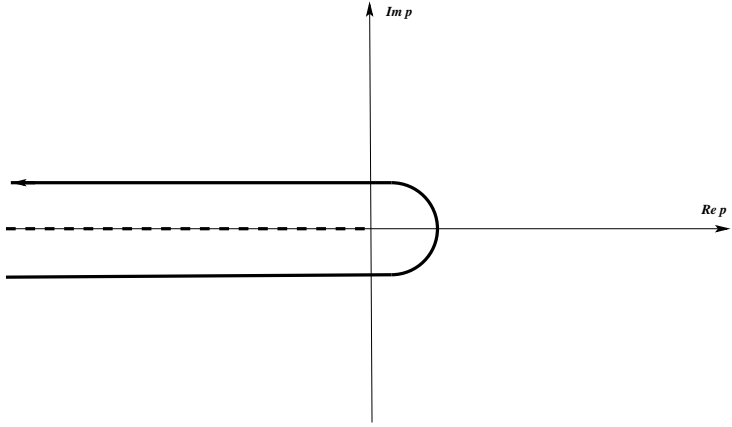}
\caption{The Hankel contour $\mathcal{H}$.} \label{fig:Hankel}
\end{figure}
Numerical methods of evaluating the Wright functions can be found in \cite{LuchkoTrujilloVelasco10} and 
\cite{LuchkoWright}.

The Mellin transform of $M_\gamma(z)$ is $\Gamma(1-s)/\Gamma(1-\gamma s)$ and
\begin{equation}
M_\gamma(z) = \frac{1}{2 \uppi \ii} 
\int_{-\varepsilon-\ii\infty}^{-\varepsilon+\ii\infty} \frac{\Gamma(1-s)}{\Gamma(1-\gamma s)} 
\, z^{-s} \, \dd s \quad \text{for $z > 0$, $0 < \gamma < 1$, $0 < \varepsilon < 1$}
\end{equation}
As a rough check note that for $z > 1$ the Bromwich contour can be closed in the right $s$ half-plane. The 
contribution of the poles of $\Gamma(1-s)$ at $s = n+1$, $n = 0, 1, 2,\ldots$ is
$$\sum_{n=0}^\infty (-1)^n \frac{z^n}{n!\, \Gamma(1 - (n+1)\,\gamma)} = M_\gamma(z)$$

Similarly, the Mellin transform of $N_\gamma(z)$ is $\Gamma(1-s)/\Gamma(2-\gamma s)$.

\section{Proof of identity \eqref{eq:identity} and related results.}
\label{app:identity}

\begin{lemma} \label{lem:Ebeta}
\begin{gather} \label{eq:Ebeta}
E_\beta(-z) = \frac{1}{2 \uppi \ii} \int_{\mathcal{B}} \frac{\Gamma(s)\, \Gamma(1-s)}{\Gamma(1 - \beta s)} z^{-s}\, 
\dd s\\
E_{\beta,2}(-z) = \frac{1}{2 \uppi \ii} \int_{\mathcal{B}} \frac{\Gamma(s)\, \Gamma(1-s)}{\Gamma(2 - \beta s)} z^{-s}\, 
\dd s  \label{eq:E1beta}
\end{gather}
\end{lemma}
\begin{proof}
Since $\Gamma(s)\, \Gamma(1-s) = \uppi/\sin(\uppi s)$, the integrand on the right-hand side of 
equation~\eqref{eq:Ebeta} has poles at $s = n = 0, \pm1, \pm 2,\ldots$. Stirling's formula
$$\Gamma(z) = \sqrt{2 \uppi} \, z^{-1/2}\, (z/\e)^z \, (1 + \mathrm{O}[1/z]) 
\qquad\text{for $\vert \arg(z) \vert < \uppi - \varepsilon$, $\varepsilon > 0$} $$ implies that for 
$\re s \rightarrow 
-\infty$
\begin{multline*}
z^{-s}/\Gamma(1 - \beta \, s) \sim_{s\rightarrow\infty} (-\beta s)^{1/2}/\sqrt{2 \uppi} \,
\exp\left( (\beta s - 1) \, [\ln(1 - \beta s) - 1] - s \, \ln(z) \right) 
\end{multline*}
The first term in the exponent dominates and therefore the integrand vanishes for $\re s \rightarrow -\infty$ faster than 
$1/\vert s \vert$. 
Hence the Bromwich contour can be closed by a half-circle at infinity in the left half of the complex $s$-plane without changing the value of the integral. The only singularities inside the closed contour thus obtained are the poles at
$s = 0, -1, -2,  $. Their residues are 
$(-z)^n/\Gamma(1+\beta n)$. Hence the right-hand side of equation~\eqref{eq:Ebeta} is equal to
$$\sum_{n=0}^\infty \frac{(-z)^n}{\Gamma(1 + n\, \beta)} \equiv E_\beta(-z)$$

The proof of equation~\eqref{eq:E1beta} is analogous.
\end{proof}
The integrands on the right hand side of \eqref{eq:Ebeta} and \eqref{eq:E1beta} are thus Mellin transforms of the left-hand sides:
\begin{gather}
\int_0^\infty E_\beta(-x) \, x^{s-1} \, \dd s = \frac{\Gamma(s)\, \Gamma(1-s)}{\Gamma(1 - \beta s)}\\
\int_0^\infty E_{\beta,2}(-x) \, x^{s-1} \, \dd s = \frac{\Gamma(s)\, \Gamma(1-s)}{\Gamma(2 - \beta s)}
\end{gather}

\begin{lemma}
\begin{gather} \label{eq:Mgamma}
M_\gamma(y) = \frac{1}{2 \uppi \ii} \int_{\mathcal{B}} \frac{\Gamma(s)}{\Gamma(1 + \gamma (s - 1))}  y^{-s} \, \dd s \qquad 0 < \gamma < 1\\
N_\gamma(y) = \frac{1}{2 \uppi \ii} \int_{\mathcal{B}} \frac{\Gamma(s)}{\Gamma(2 + \gamma (s - 1))}   y^{-s} \, \dd s\qquad 0 < \gamma < 1
\end{gather}
\end{lemma}
\begin{proof}
The proof is similar to the proof of Lemma~\ref{lem:Ebeta}. The residues of $\Gamma(s)$ at $s = -n$, $n = 0, 1,\ldots$.
are $(-1)^n/n!$, hence the integral on the right-hand side of \eqref{eq:Mgamma} is equal to 
$$\sum_{n=0}^\infty \frac{z^n}{n! \, \Gamma(1 - \gamma (n + 1))}, $$ which is equal to the function on the 
left-hand side.

The other equation is proved in the same way.
\end{proof}

We are now ready to prove the following important relations
\begin{gather}  \label{eq:important1}
E_\beta\left(-\kappa^\alpha\right) = \int_0^\infty M_\gamma(\xi) \, E_\alpha\left(-(\kappa \xi)^\alpha\right)\, \dd \xi\\
E_{\beta,2}\left(-\kappa^\alpha\right) = \int_0^\infty N_\gamma(\xi) \, 
E_\alpha\left(-(\kappa \xi)^\alpha\right) \, \dd \xi \label{eq:important2}
\end{gather}
where $0 < \gamma = \beta/\alpha < 1$.

Equation~\eqref{eq:Ebeta} implies that
\begin{equation}
E_\beta(-\kappa^\alpha) = \frac{1}{2 \uppi \ii \alpha} \int_{\mathcal{B}} 
\frac{\Gamma(s/\alpha)\, \Gamma(1 - s/\alpha)}{\Gamma(1 - \gamma s) } \kappa^{-s}\, \dd s
\end{equation}
while equation \eqref{eq:Mgamma} implies that
\begin{equation}
\frac{1}{\xi} M_\gamma(1/\xi) = \frac{1}{2 \uppi \ii} \int_{\mathcal{B}} \frac{\Gamma(1 - s)}{\Gamma(1 - \gamma s)} \xi^{-s} \, \dd s 
\end{equation}
We now recall that the Mellin transform of the Mellin convolution
$$\int_0^\infty f(\xi)\, g(x/\xi) \frac{\dd \xi}{\xi} $$
of two functions $f$ and $g$ is equal to the product of the Mellin transforms of $f$ and $g$. Hence
$$
E_\beta\left(-\kappa^\alpha\right) = \int_0^\infty \frac{1}{\xi} M_\gamma(1/\xi)\, E_\alpha\left(-(\kappa/\xi)^\alpha\right) \frac{\dd \xi}{\xi}  = 
\int_0^\infty M_\gamma(\zeta) \, E_\alpha\left( -(\kappa \zeta)^\alpha\right) \, \dd \zeta
$$
This proves equation~\eqref{eq:important1}. Equation~\eqref{eq:important2} is proved in a similar way.

The identities \eqref{eq:important1} and \eqref{eq:important2} have an important corollary. Applying the inverse Fourier transform 
to both sides of either identity and noting \eqref{eq:Ftr} we obtain the relations
\begin{gather}  \label{eq:Gconv}
G_{\beta,\alpha}(1,x) = \int_0^\infty M_\gamma(\xi) \, G_{\alpha,\alpha}(1,x/\xi)\, \frac{\dd \xi}{\xi}\\
H_{\beta,\alpha}(1,x) = \int_0^\infty N_\gamma(\xi) \, G_{\alpha,\alpha}(1,x/\xi)\, \frac{\dd \xi}{\xi}
\label{eq:Hconv}
\end{gather}

\section{The functions $M_{2/3}(z)$ and $N_{2/3}(z)$.}
\label{app:twothirds}

A method for numerical computation of the Wright function can be found in \cite{LuchkoTrujilloVelasco10,LuchkoWright}. 

Explicit representations 
of $M_\gamma$ exist for $\gamma = 1/3, 1/2$ and $2/3$. The first two can be found in 
\cite{MainardiMLWright2011}. Since we are interested in wave equations, only $\gamma > 1/2$ is of interest to us.

The function $M_{2/3}(z)$ can be expressed in terms of the Airy function using the results obtained in 
\cite{HanQAM} for the function
\begin{equation}
f_1^{(3)}(t,\lambda_1,\lambda_2) = -\frac{\partial}{\partial \lambda_2} f_2^{(3)}(t,\lambda_1,\lambda_2)
\end{equation}
where
\begin{equation} \label{eq:defF}
f_2^{(3)}(t,\lambda_1,\lambda_2) := \frac{1}{2 \uppi \ii} \int_{\mathcal{B}} s^{-2/3}\, 
\e^{s t} \, \e^{-\lambda_1\, s^{2/3} - \lambda_2 s^{1/3}} \, \dd s 
\end{equation}
and $\mathcal{B}$ denotes the Bromwich contour.
The function $f^{(3)}_2$  can be expressed in terms of the Airy function as follows:
\begin{equation} \label{eq:f23}
f_2^{(3)}(t,\lambda_1,\lambda_2) = \frac{3^{2/3}}{t^{1/3}} \Ai\left((3 t)^{-1/3}\,
\left(\lambda_2 + \frac{\lambda_1^{;2}}{3 t}\right)\right) \, \exp\left( -\frac{\lambda_1}{3 t} \left(\lambda_2 + 
\frac{2 \lambda_1^{;2}}{9 t}\right)\right)
\end{equation}

On the other hand straightforward transformations of the integration variable in \eqref{eq:defF}
yield the relation
\begin{equation} \label{eq:M23a}
M_{2/3}(z) = f_1^{(3)}\left(z^{-3/2},1,0\right)/z
\end{equation}
It is to be noted that the Hankel contour in \eqref{eq:M} can be replaced by the Bromwich contour.

Equation~\eqref{eq:M23a} can be worked out in terms of the Airy function
\begin{equation}
M_{2/3}(z) = \left[3^{-1/3} \, z \, \Ai\left(z^2/3^{4/3}\right) - 
3^{1/3} \Ai^\prime\left(z^2/3^{4/3}\right)\right]\, \e^{-2 z^3/27}
\end{equation}
or in terms of the modified Bessel function of the second kind
\begin{equation} \label{eq:M23K}
M_{2/3}(z) = \frac{z^2}{3^{3/2}\, \uppi} \left[K_{1/3}\left(\frac{2}{27} z^3\right) + 
K_{2/3}\left(\frac{2}{27} z^3\right) \right]\,
\e^{-2 z^3/27}
\end{equation}

A simple expression for the function $U_1^{(2/3,\alpha)}(y)$ can be derived from equation~\eqref{eq:M23K}:
\begin{equation}
U_1^{(2/3,\alpha)}(y) = \frac{\sqrt{3}}{2 \uppi} \int_0^\infty Z(\zeta) \, 
X_\alpha\left(2^{1/3}\, y/\left(3\, \zeta^{1/3}\right)\right) \,\dd \zeta
\end{equation} 
where 
\begin{equation}
Z(\zeta) := \left( K_{1/3}(\zeta) + K_{2/3}(\zeta) \right)\, \e^{-\zeta}
\end{equation}

In order to express $N_{2/3}$ in terms of the Airy functions, we define a new function
$$
F(t,\lambda_1,\lambda_2) := \int_0^{\lambda_1} f_2^{(3)}(t,\xi,\lambda_2)\, \dd \xi 
\equiv \frac{1}{2 \uppi \ii} \int_{\mathcal{B}} s^{-4/3}\, 
\e^{s t} \, \e^{-\lambda_1\, s^{2/3} - \lambda_2 s^{1/3}} \, \dd s
$$
In terms of this function
$$N_{3/2}(z) = z^{1/2} \int_0^1 F\left(z^{-3/2},\xi,0\right) \, \dd \xi$$
Substituting \eqref{eq:f23} the following alternative formula is derived
\begin{equation}
N_{2/3}(z) = 3^{2/3} \, z^{-1} \int_0^z \zeta^{1/2}\, \Ai\left(3^{-4/3}\, \zeta^2\right) \, \exp\left(-\frac{2}{27} \,\zeta^3\right) \, \dd \zeta
\end{equation}
$N_{2/3}$ can be further simplified by substituting 
$$\Ai(x) = \frac{1}{\uppi} 3^{-1/2}\, x^{1/2}\, K_{1/3}\left(\frac{2}{3} x^{3/2}\right)$$
which yields the formula
\begin{equation}
N_{2/3}(z) = \frac{3}{2^{5/6}\, \uppi z} \int_0^{2 z^3/27}  y^{-1/6}\, K_{1/3}(y) \, \e^{-y} \, \dd y
\end{equation}
or
\begin{multline}
N_{2/3}(z) = \frac{z^{1/2}}{\uppi \sqrt{3}} [ \Gamma(1/3) \; _2F_2\left(1/6,1/2;1/3,1/2;-4 z^3/27\right) \\
- \frac{1}{7} 
\Gamma(2/3) \, z^2 \, _2F_2\left(5/6,7/6;5/3,13/3;-4 z^3/27\right)]
\end{multline}

The functions $M_{2/3}$ and $N_{2/3}$ are shown in Fig.~\ref{fig:MN23}.
\begin{figure}
\includegraphics[width=0.8\linewidth]{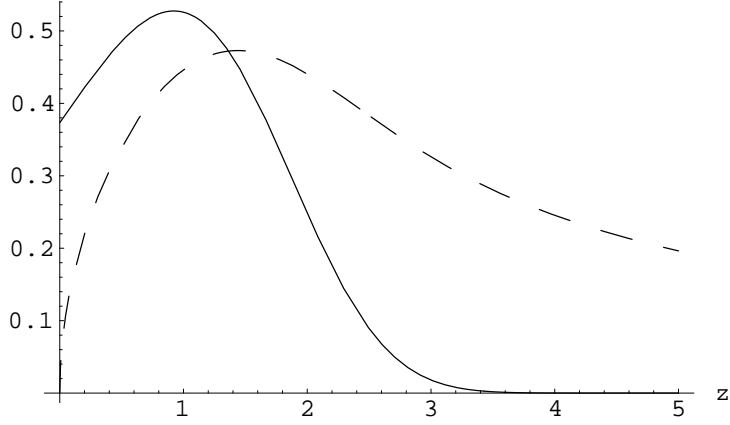}
\caption{Plots of $M_{2/3}(z)$ (solid line) and $N_{2/3}(z)$ (dashed line).}
\label{fig:MN23}
\end{figure}
 
\section{The inverse Fourier transform of $E_\alpha\left(-\vert k \vert^\alpha\right)$.}
\label{app:X}

Applying the inverse Fourier transformation to both sides of the identity 
$$ E_\alpha\left(-\vert k \vert^\alpha\right) = 
\sum_{n=0}^\infty (-1)^n \, \frac{\vert k \vert^{\alpha n}}{\Gamma(1 + \alpha n)}$$
and recalling the Fourier transform of $\vert k\vert^\alpha$ is 
$-2 \sin(n \alpha \uppi/2) \, \Gamma(1 + \alpha n)\, \vert x \vert^{-\alpha - 1}$ \cite{Gelfand}
we have
\begin{multline*}
\frac{1}{2 \uppi} \int_{-\infty}^\infty \e^{\ii k x}\, E_\alpha\left(-\vert k \vert^\alpha\right) \, \dd k = 
-2 \im \sum_{n=0}^\infty (-1)^n \, \vert x \vert^{-\alpha n - 1} \, \left(\e^{\ii \alpha \uppi/2}\right)^n = \\
-\frac{2}{\vert x \vert} \im \frac{1}{1 + \vert x \vert^{-\alpha} \, \e^{\ii \uppi \alpha/2}}
\end{multline*}
so that
\begin{multline} \label{eq:X}
G_{\alpha,\alpha}(1,x) \equiv X_\alpha(x) = \frac{1}{2 \uppi \ii} \int_{-\infty}^\infty \e^{\ii k x} \, E_{\alpha,\alpha}\left(-\vert k\vert^\alpha\right) \, \dd k = \\
\frac{\sin(\uppi \alpha/2)}{\uppi} \frac{\vert x \vert^{\alpha-1}}{\vert x \vert^{2 \alpha} + 2 \vert x \vert^\alpha \, \cos(\uppi \alpha/2) + 1}
\end{multline}
\cite{MainardiLuchkoPagnini}.

The second Green's function $H_{\alpha,\alpha}$ is somewhat more difficult to calculate
$$
H_{\alpha,\alpha}(1,x) \equiv Y_\alpha(x) = \frac{1}{2 \uppi} \int_{-\infty}^\infty \e^{\ii k x}\, E_{\alpha,2}\left(-\vert k \vert^\alpha\right) \, \dd k =: F(\vert x \vert)
$$
But $$E_{\alpha,2}\left(-\kappa^\alpha\right) = \sum_{n=0}^\infty (-1)^n \, \frac{\kappa^{\alpha n}}{\Gamma(2 + \alpha n)}$$
Hence 
\begin{multline*}
F^\prime(y) = -\frac{2}{y^2} \im \sum_{n=0}^\infty (-1)^n \, 
\left(\e^{\ii \alpha n \uppi/2}\, y\right)^{-\alpha n} = \\ 2 \sin(\alpha \uppi/2)\, \frac{y^{\alpha-2}}{y^{2 \alpha} + 2
y^\alpha \, \cos(\alpha \uppi/2) + 1}
\end{multline*}
and $F(y)$ vanishes at infinity. Hence
\begin{equation} \label{eq:Y}
F(y) = -2 \sin(\alpha \uppi/2) \int_y^\infty \frac{z^{\alpha-2}}{z^{2 \alpha} + 2
z^\alpha \, \cos(\alpha \uppi/2) + 1} \dd z
\end{equation} 

Unlike $G_{\alpha,\alpha}$, the function $H_{\alpha,\alpha}$ is not an algebraic function. For $\alpha = 1/2$
it can be expressed in terms of the logarithm.

Combining equations \eqref{eq:X}, \eqref{eq:Y}, \eqref{eq:Gconv}, \eqref{eq:Hconv} and \eqref{eq:1Dto3D} we 
get fairly explicit integral representations of the fundamental solutions of the Cauchy problem in one and 
three dimensions. 

\end{document}